\documentclass[journal,twoside,print]{ieeecolor}
\usepackage{generic}
\usepackage{cite}
\usepackage{bbding}
\usepackage{wrapfig}
\usepackage{pifont}
\usepackage{amsmath,amssymb,amsfonts}
\usepackage{algorithmic}
\usepackage{graphicx}
\usepackage{algorithm,algorithmic}
\usepackage{url}
\usepackage[justification=centering]{caption}
\usepackage[T1]{fontenc}
\usepackage{bm}
\usepackage{cases}
\usepackage{textcomp}
\usepackage[hidelinks]{hyperref}

\newtheorem{assumption}{Assumption}
\newtheorem{theorem}{Theorem}
\newtheorem{problem}{Problem}

\newtheorem{lemma}{Lemma}
\newtheorem{remark}{Remark}

\newtheorem{proposition}{Proposition}
\newtheorem{example}{Example}

\begin{document}
\title{Sufficient and Necessary Smooth Barrier-like Conditions for Continuous-Time Stochastic Reach-Avoid Verification \thanks{This work was supported by the National Research Foundation, Singapore, under its RSS Scheme (NRF-RSS2022-009) and the Basic Research Program of the Institute of Software, Chinese Academy of Sciences (Grant No. ISCAS-JCMS-202302).}}
\author{Bai Xue\\
\small KLSS, Institute of Software, Chinese Academy of Sciences, Beijing, China\\
Email: xuebai@ios.ac.cn 
} 

\maketitle

\begin{abstract}
In this paper, we study infinite-horizon reach-avoid verification for continuous-time stochastic systems modeled by stochastic differential equations (SDEs). Given a safe set, a target set, and an initial set, the goal is to certify whether the reach-avoid probability that the system, starting from the initial set, eventually reaches the target set while remaining within the safe set until the target is reached is at least a specified threshold. We study this problem within a barrier-function-based framework, which transforms the verification problem into an existence problem for barrier functions satisfying barrier-like conditions expressed as functional inequalities. We provide sufficient and necessary barrier-like conditions in terms of polynomial barrier functions for infinite-horizon reach-avoid verification under suitable regularity and uniform ellipticity assumptions. We first construct a discounted value function that characterizes lower bounds on the reach-avoid probability. We then show that it is the unique classical solution of an associated elliptic Dirichlet problem. Based on this characterization, we further show that the barrier-like condition proposed in our previous work on finite-horizon reach-avoid verification is not only sufficient for infinite-horizon reach-avoid verification but also necessary whenever the reach-avoid probability is strictly larger than the specified threshold. In particular, whenever the reach-avoid probability is strictly larger than the specified threshold fro every state in the initial set, polynomial barrier functions satisfying this barrier-like condition exist. Furthermore, when the system dynamics are polynomial, we formulate the problem of finding polynomial barrier functions satisfying this barrier-like condition as sum-of-squares (SOS) programs, which are shown to be sound and complete. Finally, we demonstrate the theoretical results on two numerical examples.
\end{abstract}


\section{Introduction}
\label{sec:int}

Stochastic differential equations (SDEs) provide a fundamental framework for modeling continuous-time dynamical systems subject to random disturbances and uncertainties \cite{oksendal2003stochastic}. They arise naturally in many safety-critical applications, including autonomous vehicles, robots, and other robotic systems. For such systems, critical safety and performance specifications often need to be verified rigorously before practical deployment. In this context, probabilistic guarantees can provide more informative assessments of system behavior than deterministic guarantees \cite{baier2008principles}.

Among the critical specifications to be verified, reach-avoid verification has attracted increasing attention in recent years due to its explicit consideration of both safety and task completion. Specifically, it considers the probability of eventually reaching a target set while satisfying prescribed safety constraints, thereby providing a natural framework for analyzing the reliability of stochastic autonomous and learning-enabled systems. Barrier-function-based approaches provide an attractive framework for reach-avoid verification because they can establish rigorous guarantees without requiring an explicit characterization of the full distribution of the underlying stochastic process, which can be difficult to obtain in practice. The key idea is to reformulate the verification problem as the existence of a barrier function satisfying functional inequalities, commonly referred to as barrier-like conditions \cite{prajna2004safety,prajna2007framework}. Various computational methods have also been developed to efficiently synthesize such barrier functions; see, e.g., \cite{ahmadi2019dsos,zhao2020synthesizing,peruffo2021automated,chatterjee2023learner}. For systems modeled by SDEs, a variety of sufficient barrier-like conditions have been developed for reach-avoid verification; see, e.g., \cite{meng2022sufficient,meng2024stochastic,xue2024,xue2026new}. Despite these developments, whether such barrier-like conditions are also necessary remains largely unexplored. This gap motivates the development of necessary and sufficient barrier-like conditions for reach-avoid verification of SDEs in this paper.

In this paper, we address this question for infinite-horizon reach-avoid verification of SDEs, where the safe set is open and bounded, and the initial and target sets are compact and contained in the interior of the safe set. Under suitable regularity, boundary smoothness, and uniform ellipticity assumptions, we establish sufficient and necessary barrier-like conditions involving polynomial barrier functions for infinite-horizon reach-avoid verification. To this end, we first establish a connection between the reach-avoid probability and an associated elliptic Dirichlet problem through a discounted reach-avoid value function. We show that this value function is the classical solution of the elliptic Dirichlet problem and provides a lower bound on the reach-avoid probability. We then show that the barrier-like condition for finite-horizon reach-avoid verification developed in \cite{xue2025refined} can also be used for infinite-horizon verification. Moreover, when the reach-avoid probability is strictly larger than the prescribed verification threshold for every state in the initial set, we construct a polynomial barrier function satisfying this condition. The construction is based on the discounted value function and its associated elliptic Dirichlet problem, thereby establishing the necessity of the barrier-like condition. When the system is polynomial, namely, each component of the drift vector and each entry of the diffusion matrix is a polynomial, and the initial, safe, and target sets are semialgebraic, we formulate the synthesis of polynomial barrier functions as an SOS program. Based on Putinar's Positivstellensatz, we establish the soundness and completeness of the resulting SOS formulation. Finally, we demonstrate the theoretical results through two numerical examples.

The main contributions of this paper are summarized below.
\begin{enumerate}
    \item We establish sufficient and necessary barrier-like conditions in terms of polynomial barrier functions  for infinite-horizon reach-avoid verification of SDEs under suitable regularity and uniform ellipticity assumptions.

    \item We develop an SOS-based computational framework for synthesizing the polynomial barrier functions. For polynomial SDEs and compact basic semialgebraic regions, Putinar's Positivstellensatz yields a complete degree hierarchy at the level of polynomial certificates.
\end{enumerate}

\subsection*{Related Work}

Barrier-function-based methods were originally developed for safety verification of deterministic dynamical systems \cite{prajna2004safety,prajna2007convex}, where barrier functions characterize regions separating the initial and unsafe sets. They have since been extended to stochastic systems to provide probabilistic guarantees for specifications such as safety and reach-avoid \cite{prajna2004stochastic,prajna2007framework,steinhardt2012finite,santoyo2021barrier,jagtap2018temporal,liu2019reachability,yu2023safe,vzikelic2023learning,xue2024sufficient,chen2025construction,laurenti2025unifying}, based either on Ville's inequality \cite{ville1939etude} or on relaxations of the underlying equations, as well as to controller synthesis \cite{wang2021safety,vzikelic2023learning,jagtap2020formal,kordabad2024control}. Related developments have also considered probabilistic programs and more general temporal specifications, including $\omega$-regular properties \cite{chakarov2013probabilistic,mciver2017new,kenyon2021supermartingales,abate2024stochastic,wang2025verifying}. Since this work focuses on continuous-time stochastic systems modeled by SDEs, we next review barrier-function-based approaches specifically developed for such systems.

Safety and reach-avoid analysis for SDEs has been associated with solving Kolmogorov equations and, in controlled settings, Hamilton–Jacobi–Bellman equations \cite{koutsoukos2008computational,bujorianu2007new,bujorianu2009stochastic,esfahani2016stochastic}. These equations are often analytically intractable and are commonly solved by discretization-based numerical methods.  In contrast, barrier-function-based methods reduce the verification problem to checking functional inequalities, which can often be certified computationally, for example via SOS programming \cite{parrilo2003semidefinite}. The pioneering works by Prajna et al. \cite{prajna2004stochastic,prajna2007framework} introduced barrier-function-based methods for infinite-horizon safety verification of SDEs, building on Ville's inequality \cite{ville1939etude,kushner1967}. These methods were subsequently extended to finite-horizon safety verification and safe controller synthesis; see, e.g., \cite{steinhardt2012finite,santoyo2021barrier,wang2021safety,feng2020unbounded}.

Compared with safety verification, where barrier functions are typically used to upper-bound the probability of reaching an unsafe set, reach-avoid verification, which seeks to lower-bound the probability of reaching a target set while avoiding an unsafe set, has received comparatively less attention. \cite{meng2024stochastic,meng2022sufficient} proposed stochastic Lyapunov-barrier functions to develop barrier-like conditions on probabilistic reach-avoid-stay specifications.  A barrier-like condition was introduced  in \cite{xue2022reach} for computing $p$-reach-avoid sets of SDEs, which can be straightforwardly extended to infinite-horizon reach-avoid verification. The barrier-like condition is derived by relaxing a system of equations whose solution characterizes the exact reach-avoid probability. This framework was subsequently extended to finite-horizon reach-avoid verification in \cite{xue2023new}, with the barrier-like condition further refined in \cite{xue2025refined}. Recently, \cite{xue2026quantitative} propose barrier-like conditions to lower and upper bound the probability that a system reaches a target over a bounded time horizon and spends a required amount of time inside a target region while obeying safety rules. 

However, existing barrier-function-based approaches for safety and reach-avoid verification of SDEs are primarily sufficient: feasibility of a barrier-like condition guarantees the desired probabilistic property, whereas infeasibility does not imply that the property fails. Recent work has made progress toward necessary conditions for discrete-time stochastic systems \cite{majumdar2025sound,kordabad2026certificates,xue2026sufficient,xue2026converse}, but the continuous-time case remains less developed. This work addresses this gap by establishing a necessary and sufficient barrier-like condition, expressed in terms of polynomial barrier functions, for infinite-horizon reach-avoid verification of continuous-time stochastic systems modeled by SDEs under suitable regularity assumptions. Moreover, we develop an SOS-based semidefinite programming formulation that is complete over polynomial barrier functions through its degree hierarchy.

\noindent\textbf{Notations.}
Throughout the paper, we use the following notation: $\mathbb{R}$ denotes the set of real numbers, $\mathbb{R}^n$ denotes the $n$-dimensional Euclidean space, and $\mathbb{R}^{n\times k}$ denotes the space of real-valued $n\times k$ matrices. For sets $\Delta_1$ and $\Delta_2$, $\partial\Delta_1$, $\overline{\Delta_1}$, and $\Delta_1\setminus\Delta_2$ denote the boundary of $\Delta_1$, the closure of $\Delta_1$, and the set difference between $\Delta_1$ and $\Delta_2$, respectively. For a set $\Delta \subseteq\mathbb R^n$, we denote by $C(\Delta)$, $C^1(\Delta)$, and $C^2(\Delta)$ the space of real-valued continuous functions on $\Delta$, the space of continuously differentiable functions on $\Delta$, and the space of twice continuously differentiable functions on $\Delta$, respectively; $\mathbb{R}[\bm{x}]$denote the ring of real-valued polynomials in the state variables $\bm{x}$.

Besides, we use Hölder spaces to quantify the regularity of
functions and domains. Let $D\subseteq\mathbb R^n$ be a domain and let
$\gamma\in(0,1)$. The space $C^{0,\gamma}(D)$ consists of all continuous
functions $v:D\to\mathbb R$ for which there exists a constant $L>0$ such
that
\[
|v(\bm x)-v(\bm y)|
\leq
L\|\bm x-\bm y\|^\gamma,
\qquad
\forall\bm x,\bm y\in D.
\]
More generally, for an integer $k\geq0$, the space $C^{k,\gamma}(D)$
consists of functions in $C^k(D)$ whose partial derivatives of order $k$
are Hölder continuous with exponent $\gamma$. In particular,
$C^{2,\gamma}(D)$ denotes the space of twice continuously differentiable
functions whose second-order partial derivatives are Hölder continuous with
exponent $\gamma$. For vector-valued functions,
$C^{0,\gamma}(D;\mathbb R^n)$ denotes the space of
$\mathbb R^n$-valued functions whose components are Hölder continuous with
exponent $\gamma$. Equivalently, $\bm v:D\to\mathbb R^n$ belongs to
$C^{0,\gamma}(D;\mathbb R^n)$ if there exists $L>0$ such that
\[
\|\bm v(\bm x)-\bm v(\bm y)\|
\leq
L\|\bm x-\bm y\|^\gamma,
\qquad
\forall\bm x,\bm y\in D.
\] 
A set $\Delta\subseteq\mathbb R^n$ is called a bounded
$C^{2,\gamma}$ domain if $\Delta$ is open, connected, and bounded, and its
boundary $\partial\Delta$ is of class $C^{2,\gamma}$. Equivalently, every
boundary point admits a neighborhood in which, after a suitable translation
and rotation of coordinates, $\partial\Delta$ can be represented as the
graph
\[
x_n=\phi(x_1,\ldots,x_{n-1}),
\]
where $\phi\in C^{2,\gamma}$. These regularity notions will be used throughout the paper to state the
assumptions on the drift, diffusion covariance, and continuation domain.

The remainder of this paper is organized as follows. Section \ref{sec:pre} introduces the SDE model, the reach-avoid problem, and the standing assumptions. Section \ref{sec:sbc} develops a sufficient and necessary barrier-like conditions involving polynomial barrier functions. Section \ref{sec:sdp} presents the SOS
formulation and the corresponding completeness result. Numerical examples
are provided in Section \ref{sec:ex}. Finally, we conclude this paper in Section \ref{sec:con}.

\section{Preliminaries}
\label{sec:pre}

In this section, we introduce continuous-time stochastic systems modeled by SDEs and formulate the infinite-horizon reach-avoid verification problem. We begin by introducing SDEs and the assumptions that ensure the well-posedness of their solutions. We then formulate the infinite-horizon reach-avoid verification problem and present the assumptions adopted throughout this paper for constructing barrier-like conditions.

We consider the continuous-time stochastic system modeled by SDEs of the following form:
\begin{equation}
\label{SDE}
d\bm{X}(t)=\bm{b}(\bm{X}(t))dt +
\bm{\sigma}(\bm{X}(t))d\bm{W}(t),
\end{equation}
where $\bm{W}=(W_1,\ldots,W_k)^\top$ is a standard $k$-dimensional
Wiener process defined on a probability space
$(\Omega,\mathcal{F},\mathbb{P})$. The expectation with respect to $\mathbb{P}$ is denoted by $\mathbb{E}[\cdot]$.

We impose the following standard assumptions on the drift and diffusion coefficients.

\begin{assumption}[Local well-posedness]
\label{assump:sde_conditions}
The drift $\bm{b}:\mathbb R^n\to\mathbb R^n$ and diffusion
$\bm{\sigma}:\mathbb R^n\to\mathbb R^{n\times k}$ are locally Lipschitz continuous; that is, for every compact set $K\subset\mathbb R^n$, there
exists $L_K>0$ such that
\[
\|\bm{b}(\bm{x})-\bm{b}(\bm{y})\|
+
\|\bm{\sigma}(\bm{x})-\bm{\sigma}(\bm{y})\|_F
\leq
L_K\|\bm{x}-\bm{y}\|, \forall\bm{x},\bm{y}\in K,
\]
where $|\cdot|_F$ denotes the Frobenius norm.
\end{assumption}

Under Assumption \ref{assump:sde_conditions}, for every initial state $\bm{x}_0\in\mathbb R^n$, the SDE \eqref{SDE} admits a unique maximal strong solution, denoted by $\bm{X}_{\bm{x}_0}(t)$. In the analysis below, we only consider this solution up to the first exit time from the bounded continuation domain introduced below. 

For a twice continuously differentiable function  $v:\mathbb{R}^n\to\mathbb{R}$, the infinitesimal generator associated with \eqref{SDE} is defined by
\begin{equation}
\label{infi}
\mathcal{L}v(\bm{x})
=
\lim_{h\downarrow0}
\frac{
\mathbb{E}_{\bm{x}}
\left[
v(\bm{X}_{\bm{x}}(h))
\right]
-v(\bm{x})
}{h},
\end{equation}
whenever the limit exists, where $\mathbb{E}_{\bm{x}}$ denotes
expectation for the solution initialized at $\bm{x}$. When the initial state is clear from the context, we simply write $\mathbb{E}$.

\begin{proposition}
\label{prop:inf_generator}
Consider the system~\eqref{SDE} under
Assumption~\ref{assump:sde_conditions}. For every
$v\in C^2(\mathbb{R}^n)$, the infinitesimal generator is given by
\begin{equation}
\label{generator}
\mathcal{L}v(\bm{x})= \nabla v(\bm{x})^\top\bm{b}(\bm{x}) +
\frac{1}{2}
\operatorname{tr}
\left(
\bm{\sigma}(\bm{x})^\top
\bm{H}v(\bm{x})
\bm{\sigma}(\bm{x})
\right),
\end{equation}
Here, $\bm{H}v(\bm{x}) = \nabla^2 v(\bm{x})$ denotes the Hessian matrix of $v$ at $\bm{x}$.
Equivalently, by defining
\[
\bm{a}(\bm{x}):= \bm{\sigma}(\bm{x})\bm{\sigma}(\bm{x})^\top,
\]
we have
\[
\mathcal{L}v(\bm{x})=\nabla v(\bm{x})^\top\bm{b}(\bm{x})+
\frac{1}{2} \operatorname{tr}
\left(
\bm{a}(\bm{x})\bm{H}v(\bm{x})
\right).
\]

Moreover, whenever $v$ satisfies the regularity and integrability
conditions required for Dynkin's formula, then for every almost surely finite stopping time $\tau$,
\begin{equation}
\label{dynkin}
\mathbb{E}
\left[
v(\bm{X}_{\bm{x}_0}(\tau))
\right]
=
v(\bm{x}_0)+ \mathbb{E}
\left[
\int_0^\tau
\mathcal{L}v(\bm{X}_{\bm{x}_0}(s))ds
\right].
\end{equation}
\end{proposition}

See, e.g.,~\cite{oksendal2013stochastic}.

We study an infinite-horizon reach-avoid verification problem.
Let $\mathcal{X}\subset\mathbb{R}^n$ be a safe set, $\mathcal{X}_r\subset\mathcal{X}$ a  target set, and
$\mathcal{X}_0\subset\mathcal{X}\setminus\mathcal{X}_r$ a set of initial states. For $\bm{x}_0\in\mathcal{X}\setminus\mathcal{X}_r$,
define the target hitting time and the safe-set exit time by
\[
\tau_r(\bm{x}_0):=
\inf\left\{
t\geq 0 \mid 
\bm{X}_{\bm{x}_0}(t)\in\mathcal{X}_r
\right\},
\]
and
\[
\tau_{\mathcal{X}}(\bm{x}_0):=
\inf\left\{
t\geq0:
\bm{X}_{\bm{x}_0}(t)\notin\mathcal{X}
\right\}.
\]
The corresponding infinite-horizon reach-avoid probability is
\begin{equation}
\label{eq:ra_probability}
\mathbb{P}_{\mathrm{RA}}(\bm{x}_0):=
\mathbb{P}
\left(
\tau_r(\bm{x}_0)< \tau_{\mathcal{X}}(\bm{x}_0)
\right).
\end{equation}

\begin{problem}[Infinite-horizon Reach-avoid Verification]
\label{un_proII}
Given a prescribed probability threshold $\epsilon \in (0,1)$
the infinite-horizon reach-avoid verification problem is to establish 
\[
\mathbb{P}_{\mathrm{RA}}(\bm{x}_0)\geq\epsilon,
\qquad
\forall\bm{x}_0\in\mathcal{X}_0.
\]
\end{problem}

We study a necessary and sufficient barrier-like condition of Problem \ref{un_proII} in terms of polynomial barrier functions. The following assumptions specify the system and set-theoretic conditions used throughout the paper.

\begin{assumption}
\label{assm2}
\begin{enumerate}
\item The safe set $\mathcal{X}$ is open and bounded, the target set $\mathcal{X}_r\subset\mathcal{X}$ is compact with non-empty interior, implying $\partial\mathcal{X}\cap\partial\mathcal{X}_r=\emptyset$, and the initial set $\mathcal{X}_0\subset\mathcal{X}\setminus\mathcal{X}_r$ is compact. The continuation set
\[
\mathcal{D}:=\mathcal{X}\setminus\mathcal{X}_r
\]
is bounded and open, and $\partial\mathcal{D}= \partial\mathcal{X}\cup\partial\mathcal{X}_r$ is a compact $C^{2,\gamma}$ hypersurface. Moreover, $\mathcal{D}$ has
finitely many connected components
\[
\mathcal{D}=\bigcup_{\ell=1}^{N}\mathcal{D}_\ell,
\]
whose closures are pairwise disjoint:
\[
\overline{\mathcal{D}_\ell}\cap\overline{\mathcal{D}_j}
=
\emptyset,
\qquad \ell\neq j.
\]

    \item The diffusion matrix $\bm{a}$ is uniformly elliptic on $\overline{\mathcal{D}}$: there exists $\lambda>0$
    such that
    \[
    \bm{\xi}^\top
    \bm{a}(\bm{x})
    \bm{\xi}
    \geq
    \lambda\|\bm{\xi}\|^2,
    \qquad
    \forall\bm{x}\in\overline{\mathcal{D}},
    \quad
    \forall\bm{\xi}\in\mathbb{R}^n.
    \]
   This is a uniform non-degeneracy condition on the diffusion coefficient.
\end{enumerate}
\end{assumption}

\section{Necessary and Sufficient Barrier-like Conditions}
\label{sec:sbc}
In this section, we introduce necessary and sufficient barrier-like conditions in terms of polynomial barrier functions for infinite-horizon reach-avoid verification. Our starting point is a discounted reach-avoid value function, which provides a lower bound on the reach-avoid probability and solves the Dirichlet problem associated with the generator $\mathcal{L}$. We then show that the barrier-like condition, originally proposed in \cite{xue2025refined} for finite-horizon reach-avoid verification, is necessary and sufficient for infinite-horizon reach-avoid
verification when the reach-avoid probability is strictly larger than $\epsilon$ for every state $\bm{x}_0\in \mathcal{X}_0$, i.e., $\mathbb{P}_{\mathrm{RA}}(\bm{x}_0)>\epsilon, \forall \bm{x}_0\in \mathcal{X}_0$.

 We begin by introducing the exit time and the discounted value function, following the construction developed for infinite-horizon reach-avoid verification of stochastic discrete-time systems in \cite{xue2026sufficient,xue2026converse}. For $\bm{x}\in\mathcal{D}$, define
\[
\tau_{\mathcal{D}}:=\inf\left\{
t\geq 0\mid 
\bm{X}_{\bm{x}}(t)\notin\mathcal{D}
\right\}.
\]
For a fixed $\alpha>0$, define the boundary function
\[
g(\bm{x})=\begin{cases}
1, & \bm{x}\in\partial\mathcal{X}_r,\\
0, & \bm{x}\in\partial\mathcal{X},
\end{cases}
\]
and the discounted reach-avoid value function
\[
v_{\alpha}(\bm{x}):=
\mathbb{E}
\left[
e^{-\alpha\tau_{\mathcal{D}}}
g\!\left(
\bm{X}_{\bm{x}}(\tau_{\mathcal{D}})
\right)
\mathbf{1}_{\{\tau_{\mathcal{D}}<\infty\}}
\right],
\qquad
\bm{x}\in\mathcal{D}.
\]
Since the sample paths are continuous,
\[
v_{\alpha}(\bm{x})=\mathbb{E}
\left[
e^{-\alpha\tau_r}
\mathbf{1}_{\{\tau_r<\tau_{\mathcal{X}}\}}
\right],
\qquad
\forall \bm{x}\in\mathcal{D}.
\]

 In the following, our first result characterizes $v_\alpha$ as the unique classical solution of the associated Dirichlet problem and provides the maximum-principle bounds that will be used repeatedly below.

\begin{lemma}
\label{diri}
Under Assumptions \ref{assump:sde_conditions} and \ref{assm2},
the function $v_{\alpha}$ belongs to
$C^{2,\gamma}(\overline{\mathcal{D}})$ and is the unique classical solution of
\begin{equation}
\label{d_equation}
\begin{cases}
\mathcal{L}u-\alpha u=0,
& \bm{x}\in\mathcal{D},\\
u=1,
& \bm{x}\in\partial\mathcal{X}_r,\\
u=0,
& \bm{x}\in\partial\mathcal{X},
\end{cases}
\end{equation}
where
\[
\mathcal{L}u(\bm{x})=
\nabla_{\bm{x}}u(\bm{x})^{\top}\bm{b}(\bm{x})+
\frac{1}{2}
\operatorname{tr}
\left(
\bm{a}(\bm{x})\bm{H}u(\bm{x})
\right).
\]
Moreover,
\[
0\leq v_{\alpha}(\bm{x})\leq1,
\qquad
\forall \bm{x}\in\overline{\mathcal{D}}.
\]
\end{lemma}

\begin{proof}
We first establish the existence and uniqueness of the classical
solution of the Dirichlet problem.

The boundary data
\[
g(\bm x)
=
\begin{cases}
1, & \bm x\in\partial\mathcal{X}_r,\\
0, & \bm x\in\partial\mathcal{X}
\end{cases}
\]
are constant on each of the two disjoint boundary parts. Hence
\[
g\in C^{2,\gamma}(\partial\mathcal{D}).
\]
Moreover, since $\partial\mathcal{X}_r$ and $\partial\mathcal{X}$ are
disjoint compact sets, $g$ admits a $C^{2,\gamma}$ extension to an open
neighborhood of $\partial\mathcal{D}$.

Let $\mathcal{D}_\ell$ be any connected component of $\mathcal{D}$.
By Assumption \ref{assm2}, $\mathcal{D}$ has finitely many connected components with pairwise disjoint closures. Since $\mathcal{D}$ is bounded, each $\mathcal{D}_\ell$ is bounded. Moreover,
\[
\partial\mathcal{D}=\bigsqcup_{\ell=1}^{N}\partial\mathcal{D}_\ell,
\]
where $\bigsqcup$ denotes a disjoint union. Hence, each
$\partial\mathcal{D}_\ell$ is a union of connected components of the compact $C^{2,\gamma}$ hypersurface $\partial\mathcal{D}$ and is
therefore itself a $C^{2,\gamma}$ hypersurface. Thus, each
$\mathcal{D}_\ell$ is a bounded and open domain with $C^{2,\gamma}$ boundary.

 Then, under Assumptions \ref{assump:sde_conditions} and \ref{assm2}, Theorem 6.14 in \cite{gilbarg1977elliptic} implies that, for every $\alpha>0$, the Dirichlet problem
\[
\begin{cases}
(\mathcal{L}-\alpha)u_\ell=0,
& \text{in }\mathcal{D}_\ell,\\[2mm]
u_\ell=g,
& \text{on }\partial\mathcal{D}_\ell
\end{cases}
\]
admits a unique solution
\[
u_\ell\in C^{2,\gamma}(\overline{\mathcal{D}_\ell}).
\]
Indeed, in the standard form of a uniformly elliptic operator, the
second-order coefficient matrix is $\frac{1}{2}\bm{a}$, the first-order coefficient vector is $\bm{b}$, and the zero-order coefficient of $\mathcal{L}-\alpha$ is $-\alpha<0$. Hence, the coefficient regularity, uniform ellipticity, boundary regularity, and boundary-data regularity required by Theorem~6.14 are satisfied. Uniqueness follows from the
maximum principle.

Solving the Dirichlet problem independently on each connected component and combining the resulting solutions gives
\[
u(\bm{x})=u_\ell(\bm{x}),
\qquad
\bm{x}\in\overline{\mathcal{D}_\ell},
\]
which is well defined because the closures
$\overline{\mathcal{D}_\ell}$ are pairwise disjoint. Since
\[
\overline{\mathcal{D}}=\bigcup_{\ell=1}^{N}\overline{\mathcal{D}_\ell},
\]
the resulting function satisfies
\[
u\in C^{2,\gamma}(\overline{\mathcal{D}})
\]
and
\[
\begin{cases}
(\mathcal{L}-\alpha)u=0,
& \text{in }\mathcal{D},\\[2mm]
u=g,
& \text{on }\partial\mathcal{D}.
\end{cases}
\]

We next prove that
\[
0\leq u\leq1
\qquad\text{on }\overline{\mathcal{D}}.
\]

Suppose first that $u$ takes a negative value. Since $u$ is continuous
on the compact set $\overline{\mathcal{D}}$, it attains its minimum.
Because $u=g\geq0$ on $\partial\mathcal{D}$, the minimum must be
attained at an interior point $\bm{x}_*\in\mathcal{D}$. Hence
\[
\nabla u(\bm{x}_*)=0,
\qquad
\bm{H}u(\bm{x}_*)\succeq0.
\]
Since $\bm{a}(\bm{x}_*)$ is positive semidefinite,
\[
\operatorname{tr}
\left(
\bm{a}(\bm{x}_*)\bm{H}u(\bm{x}_*)
\right)
\geq0,
\]
and consequently
\[
\mathcal{L}u(\bm{x}_*)\geq0.
\]
On the other hand, \eqref{d_equation} gives
\[
\mathcal{L}u(\bm{x}_*)
=
\alpha u(\bm{x}_*)
<
0,
\]
which is a contradiction. Hence
\[
u\geq0
\qquad\text{on }\overline{\mathcal{D}}.
\]

Next, suppose that $u$ takes a value larger than $1$. Since
$u\leq1$ on $\partial\mathcal{D}$, the maximum of $u$ must then be
attained at an interior point $\bm{x}^*\in\mathcal{D}$. At such a point,
\[
\nabla u(\bm{x}^*)=0,
\qquad
\bm{H}u(\bm{x}^*)\preceq0.
\]
Therefore,
\[
\mathcal{L}u(\bm{x}^*)\leq0.
\]
However, by \eqref{d_equation},
\[
\mathcal{L}u(\bm{x}^*)
=
\alpha u(\bm{x}^*)
>
0,
\]
which is again a contradiction. Thus
\[
u\leq1
\qquad\text{on }\overline{\mathcal{D}}.
\]

It remains to establish the probabilistic representation. For
$t\geq0$, define
\[
M_t
:=
\int_0^{t\wedge\tau_{\mathcal{D}}}
e^{-\alpha s}
\nabla u\!\left(
\bm{X}_{\bm{x}}(s)
\right)^{\top}
\bm{\sigma}\!\left(
\bm{X}_{\bm{x}}(s)
\right)
\,d\bm{W}_s.
\]

Because $\mathcal{D}$ is bounded and $\bm{a}$ is uniformly elliptic on $\overline{\mathcal{D}}$, $\tau_{\mathcal{D}}<\infty$ almost surely according to Proposition 10.1 in \cite{baldi2017stochastic}.

Applying It\^o's formula to
\[
e^{-\alpha(t\wedge\tau_{\mathcal{D}})}
u\!\left(
\bm{X}_{\bm{x}}(t\wedge\tau_{\mathcal{D}})
\right)
\]
gives
\[
\begin{aligned}
&e^{-\alpha(t\wedge\tau_{\mathcal{D}})}
u\!\left(
\bm{X}_{\bm{x}}(t\wedge\tau_{\mathcal{D}})
\right)\\
&\qquad=
u(\bm{x})
+
M_t
+
\int_0^{t\wedge\tau_{\mathcal{D}}}
e^{-\alpha s}
\bigl(
\mathcal{L}u-\alpha u
\bigr)
\!\left(
\bm{X}_{\bm{x}}(s)
\right)\,ds.
\end{aligned}
\]
Since $\mathcal{L}u-\alpha u=0$ in $\mathcal{D}$,
\[
e^{-\alpha(t\wedge\tau_{\mathcal{D}})}
u\!\left(
\bm{X}_{\bm{x}}(t\wedge\tau_{\mathcal{D}})
\right)
=
u(\bm{x})+M_t.
\]

Taking expectations yields
\[
u(\bm{x})=\mathbb{E}
\left[
e^{-\alpha(t\wedge\tau_{\mathcal{D}})}
u\!\left(
\bm{X}_{\bm{x}}(t\wedge\tau_{\mathcal{D}})
\right)
\right].
\]
(Because $u\in C^{2,\gamma}(\overline{\mathcal{D}})$ and $\overline{\mathcal{D}}$ is compact, the gradient term $\nabla u$ and the diffusion coefficient $\bm{\sigma}$ are strictly bounded on the path interval $[0,t\wedge \tau_{\mathcal{D}}]$. This uniform boundedness guarantees that the stochastic integral $M_t$ is a true martingale, leading to $\mathbb{E}[M_t]=0$ by the optional stopping theorem.)

Since $0\leq u\leq1$,
\[
0
\leq
e^{-\alpha(t\wedge\tau_{\mathcal{D}})}
u\!\left(
\bm{X}_{\bm{x}}(t\wedge\tau_{\mathcal{D}})
\right)
\leq1.
\]
Moreover, as $t\to\infty$,
\[
e^{-\alpha(t\wedge\tau_{\mathcal{D}})}
u\!\left(
\bm{X}_{\bm{x}}(t\wedge\tau_{\mathcal{D}})
\right)
\longrightarrow
e^{-\alpha\tau_{\mathcal{D}}}
u\!\left(
\bm{X}_{\bm{x}}(\tau_{\mathcal{D}})
\right)
\mathbf{1}_{\{\tau_{\mathcal{D}}<\infty\}}
\]
almost surely. Hence, by dominated convergence,
\[
u(\bm{x})
=
\mathbb{E}
\left[
e^{-\alpha\tau_{\mathcal{D}}}
u\!\left(
\bm{X}_{\bm{x}}(\tau_{\mathcal{D}})
\right)
\mathbf{1}_{\{\tau_{\mathcal{D}}<\infty\}}
\right].
\]

Because the trajectories are continuous, on the event
$\{\tau_{\mathcal{D}}<\infty\}$,
\[
\bm{X}_{\bm{x}}(\tau_{\mathcal{D}})
\in\partial\mathcal{D}.
\]
Therefore,
\[
u\!\left(
\bm{X}_{\bm{x}}(\tau_{\mathcal{D}})
\right)
=
g\!\left(
\bm{X}_{\bm{x}}(\tau_{\mathcal{D}})
\right),
\]
and hence
\[
u(\bm{x})
=
\mathbb{E}
\left[
e^{-\alpha\tau_{\mathcal{D}}}
g\!\left(
\bm{X}_{\bm{x}}(\tau_{\mathcal{D}})
\right)
\mathbf{1}_{\{\tau_{\mathcal{D}}<\infty\}}
\right]
=
v_{\alpha}(\bm{x}).
\]
Thus $v_{\alpha}=u$, and consequently
\[
v_{\alpha}\in C^{2,\gamma}(\overline{\mathcal{D}})
\]
is the unique classical solution of \eqref{d_equation}. The bound
\[
0\leq v_{\alpha}\leq1
\]
follows from the maximum-principle argument above.
\end{proof}

Lemma~\ref{diri} shows that the discounted value function is twice continuously differentiable on $\overline{\mathcal{D}}$, with second-order partial derivatives that are Hölder continuous with exponent $\gamma$, and is bounded between $0$ and $1$. If the reach-avoid probability is strictly larger than $\epsilon$ for every state in the initial set $\mathcal{X}*0$, i.e.,
$\mathbb{P}*{\mathrm{RA}}(\bm{x}_0)>\epsilon$ for all $\bm{x}_0\in\mathcal{X}_0$, the following lemma shows that there exists a common discount factor $\alpha>0$ such that the discounted value function remains strictly above $\epsilon$ over the entire initial set.

\begin{lemma}
\label{lemma:stri}
Under Assumptions \ref{assump:sde_conditions}--\ref{assm2}, if $\mathbb{P}_{\mathrm{RA}}(\bm{x}_0)>\epsilon, \forall\bm{x}_0\in\mathcal{X}_0$
there exists a common $\alpha>0$ such that
\[
\mathbb{P}_{\mathrm{RA}}(\bm{x}_0)
\geq
v_{\alpha}(\bm{x}_0)
>
\epsilon,
\qquad
\forall \bm{x}_0\in\mathcal{X}_0.
\]
\end{lemma}
\begin{proof}
For every fixed $\bm{x}_0\in\mathcal{X}_0$,
\[
v_{\alpha}(\bm{x}_0)=
\mathbb{E}
\left[
e^{-\alpha\tau_r}
\mathbf{1}_{\{\tau_r<\tau_{\mathcal{X}}\}}
\right].
\]
As $\alpha\downarrow0$,
\[
e^{-\alpha\tau_r}
\mathbf{1}_{\{\tau_r<\tau_{\mathcal{X}}\}}
\uparrow
\mathbf{1}_{\{\tau_r<\tau_{\mathcal{X}}\}}
\qquad\text{a.s.}
\]
Hence, by the monotone convergence theorem,
\[
\lim_{\alpha\downarrow0}
v_{\alpha}(\bm{x}_0)
=
\mathbb{P}_{\mathrm{RA}}(\bm{x}_0)
>
\epsilon.
\]
Therefore, for every $\bm{x}_0\in\mathcal{X}_0$, there exists
$\alpha_{\bm{x}_0}>0$ such that
\[
v_{\alpha_{\bm{x}_0}}(\bm{x}_0)>\epsilon.
\]

By Lemma \ref{diri},
$v_{\alpha_{\bm{x}_0}}\in C(\mathcal{D})$. Hence there exists an open
neighborhood $U_{\bm{x}_0}$ of $\bm{x}_0$ such that
\[
v_{\alpha_{\bm{x}_0}}(\bm{x})>\epsilon,
\qquad
\forall \bm{x}\in U_{\bm{x}_0}.
\]
The collection
\[
\left\{
U_{\bm{x}_0} \mid \bm{x}_0\in\mathcal{X}_0
\right\}
\]
forms an open cover of the compact set $\mathcal{X}_0$. Thus, there
exist $\bm{x}_1,\ldots,\bm{x}_N\in\mathcal{X}_0$ such that
\[
\mathcal{X}_0
\subseteq
\bigcup_{i=1}^{N}U_{\bm{x}_i}.
\]
Define
\[
\alpha
:=
\min_{1\leq i\leq N}
\alpha_{\bm{x}_i}
>0.
\]
Since $v_{\alpha}(\bm{x})$ is nonincreasing in $\alpha$, we have
\[
v_{\alpha}(\bm{x})
\geq
v_{\alpha_{\bm{x}_i}}(\bm{x}),
\qquad
0<\alpha\leq\alpha_{\bm{x}_i}.
\]
Therefore,
\[
v_{\alpha}(\bm{x})
>
\epsilon,
\qquad
\forall \bm{x}\in U_{\bm{x}_i},
\quad i=1,\ldots,N.
\]
Since the sets $U_{\bm{x}_1},\ldots,U_{\bm{x}_N}$ cover
$\mathcal{X}_0$,
\[
v_{\alpha}(\bm{x}_0)>\epsilon,
\qquad
\forall \bm{x}_0\in\mathcal{X}_0.
\]

Finally,
\[
v_{\alpha}(\bm{x}_0)
=
\mathbb{E}
\left[
e^{-\alpha\tau_r}
\mathbf{1}_{\{\tau_r<\tau_{\mathcal{X}}\}}
\right]
\leq
\mathbb{P}_{\mathrm{RA}}(\bm{x}_0),
\]
and hence
\[
\mathbb{P}_{\mathrm{RA}}(\bm{x}_0)
\geq
v_{\alpha}(\bm{x}_0)
>
\epsilon,
\qquad
\forall \bm{x}_0\in\mathcal{X}_0.
\]
We complete the proof.
\end{proof}

\begin{remark}
\label{remark1}
\emph{It is observed that when $\alpha=0$, the discounted value function satisfies $v_0(\bm{x})=\mathbb{P}_{\mathrm{RA}}(\bm{x})$ for $\bm{x}\in \mathcal{D}$. It is shown in
\cite{koutsoukos2008computational} that $v_0$ is the unique viscosity solution of the partial differential equation \eqref{d_equation}. Please also refer to Theorem 9.2.14 in \cite{oksendal2013stochastic} for $v_0 \in C^2$ . However, we explicitly assume $\alpha>0$ and do not use the case $\alpha=0$. The main reason is that $\alpha=0$ does not guarantee the existence of polynomial barrier functions satisfying our barrier-like condition.} $\blacksquare$   
\end{remark}

Lemma \ref{lemma:stri} shows that there exists a common $\alpha>0$ for which the discounted value function has a strict margin above $\epsilon$ over the entire initial set. Since the discounted value function is twice continuously differentiable, it can be approximated by polynomials while simultaneously controlling the function, its gradient, and its Hessian on $\overline{\mathcal{D}}$. The following lemma, proved via Nachbin's smooth Stone--Weierstrass theorem \cite{nachbin1949algebres}, provides exactly this approximation tool.

\begin{lemma}
\label{lem:poly_approx}
Let $U\subset\mathbb{R}^n$ be open, and let
$K\subset U$ be a compact set. For every
$f\in C^2(U)$, there exists a sequence of polynomials
$p_k\in\mathbb{R}[\bm{x}]$ such that
\[
\|p_k-f\|_{C^2(K)}\longrightarrow0,
\]
where
\[
\|g\|_{C^2(K)}:=
\max\left\{
\begin{split}
&\sup_{\bm{x}\in K}|g(\bm{x})|,
\sup_{\bm{x}\in K}\|\nabla g(\bm{x})\|,\\
&\sup_{\bm{x}\in K}\|\bm{H}g(\bm{x})\|_F
\end{split}
\right\}.
\]
\end{lemma}
\begin{proof}
Let
\[
\mathcal{P}:=
\{p|_U \mid p\in\mathbb{R}[\bm{x}]\}
\subset C^2(U)
\]
be the algebra of polynomial functions restricted to $U$.
We verify the three conditions of Nachbin's smooth
Stone--Weierstrass theorem \cite{nachbin1949algebres}.

First, $\mathcal{P}$ separates points of $U$. Indeed, if
$\bm{x},\bm{y}\in U$ with $\bm{x}\neq\bm{y}$, then there exists
$i\in\{1,\ldots,n\}$ such that $x_i\neq y_i$. The coordinate function
\[
p(\bm{z})=z_i
\]
belongs to $\mathcal{P}$ and satisfies
\[
p(\bm{x})\neq p(\bm{y}).
\]

Second, $\mathcal{P}$ is nowhere vanishing, since the constant
polynomial
\[
p(\bm{z})\equiv1
\]
belongs to $\mathcal{P}$ and satisfies
\[
p(\bm{x})\neq0,
\qquad
\forall\bm{x}\in U.
\]

Third, $\mathcal{P}$ separates tangent directions. Let
$\bm{x}\in U$ and let $\bm{v}\in\mathbb{R}^n\setminus\{\bm{0}\}$.
Choose $i\in\{1,\ldots,n\}$ such that $v_i\neq0$. For the coordinate
function
\[
p(\bm{z})=z_i,
\]
we have
\[
D_{\bm{v}}p(\bm{x})
=
\nabla p(\bm{x})^\top\bm{v}
=
v_i
\neq0.
\]
Thus, the conditions of Nachbin's theorem are satisfied. Consequently, $\mathcal{P}$ is dense in $C^2(U)$ with respect to the compact-open $C^2$ topology.

Hence, for the given $f\in C^2(U)$, there exists a sequence
$p_k\in\mathbb{R}[\bm{x}]$ such that, for every compact set
$L\subset U$,
\[
\max_{|\alpha|\leq2}
\sup_{\bm{x}\in L}
\left|
D^\alpha(p_k-f)(\bm{x})
\right|
\longrightarrow0.
\]
Since $K\subset U$, taking $L=K$ gives
\[
\begin{cases}
\sup_{\bm{x}\in K} |p_k(\bm{x})-f(\bm{x})| \longrightarrow 0,\\
\sup_{\bm{x}\in K} \|\nabla p_k(\bm{x})-\nabla f(\bm{x})\| \longrightarrow 0,\\
\sup_{\bm{x}\in K}
\|\bm{H}p_k(\bm{x})-\bm{H}f(\bm{x})\|_F\longrightarrow 0.
\end{cases}
\]
Therefore,
\[
\|p_k-f\|_{C^2(K)}
\longrightarrow0.
\]
We finish the proof.
\end{proof}

We are now ready to state and prove the main result of this section: the barrier-like condition in \cite{xue2025refined} is not only sufficient for infinite-horizon reach-avoid verification, but also necessary. Specially, when the reach-avoid probability is strictly larger than $\epsilon$ for every $\bm{x}_0\in \mathcal{X}_0$, there exists a polynomial barrier function satisfying this barrier-like condition. The proof for this conclusion combines the equation \eqref{d_equation}
 Lemma \ref{diri}, the uniform strict margin of
Lemma \ref{lemma:stri}, and the polynomial approximation of
Lemma \ref{lem:poly_approx}. 

\begin{theorem}
\label{thm:poly_sbc}
Under Assumptions \ref{assump:sde_conditions} and \ref{assm2}, if there exist
$v\in\mathbb{R}[\bm{x}]$, $\alpha>0$, and $\beta\geq0$ satisfying
\begin{equation}
\label{barrier2}
\begin{cases}
v(\bm{x})\geq
\epsilon+\dfrac{\beta}{\alpha}(\epsilon-1),
&\forall \bm{x}\in\mathcal{X}_0,\\[2mm]
\mathcal{L}v(\bm{x})-\alpha v(\bm{x})\geq\beta,
&\forall \bm{x}\in\overline{\mathcal{D}},\\[2mm]
v(\bm{x})\leq1,
&\forall \bm{x}\in\partial\mathcal{X}_r,\\[2mm]
\alpha v(\bm{x})+\beta\leq0,
&\forall \bm{x}\in\partial\mathcal{X},
\end{cases}
\end{equation}
then
\[
\mathbb{P}_{\mathrm{RA}}(\bm{x}_0)\geq\epsilon,
\qquad
\forall \bm{x}_0\in\mathcal{X}_0.
\]
Moreover, if the reach-avoid probability is strictly larger than $\epsilon$ for every state in the initial set $\mathcal{X}_0$, i.e., $\mathbb{P}_{\mathrm{RA}}(\bm{x}_0)>\epsilon, \forall\bm{x}_0\in\mathcal{X}_0$, there exist $v\in\mathbb{R}[\bm{x}]$, $\alpha>0$, and $\beta\geq0$ such that all four inequalities in \eqref{barrier2} hold strictly.
\end{theorem}

\begin{proof}
\noindent{\textbf{1. We first prove sufficiency.}}

By Corollary~1 in \cite{xue2025refined},
for every $\bm{x}_0\in\mathcal{X}_0$,
\[
\mathbb{P}_{\mathrm{RA}}(\bm{x}_0)
\geq
\lim_{T\to\infty}
\frac{
e^{\alpha T}
\bigl(
\alpha v(\bm{x}_0)+\beta
\bigr)
-(\alpha+\beta)
}{
(\alpha+\beta)(e^{\alpha T}-1)
}
=
\frac{
\alpha v(\bm{x}_0)+\beta
}{
\alpha+\beta
}.
\]
By the first condition in \eqref{barrier2},
\[
v(\bm{x}_0)
\geq
\epsilon+\frac{\beta}{\alpha}(\epsilon-1),
\]
and hence
\[
\alpha v(\bm{x}_0)+\beta
\geq
\alpha\epsilon+\beta(\epsilon-1)+\beta
=
\epsilon(\alpha+\beta).
\]
Therefore,
\[
\mathbb{P}_{\mathrm{RA}}(\bm{x}_0)\geq\epsilon,
\qquad
\forall \bm{x}_0\in\mathcal{X}_0.
\]

\noindent\textbf{2. We next prove necessity.}

By Lemma~\ref{lemma:stri}, there exists a
common $\alpha>0$ such that
\[
v_{\alpha}(\bm{x}_0)>\epsilon,
\qquad
\forall \bm{x}_0\in\mathcal{X}_0.
\]
By Lemma \ref{diri},
\[
v_{\alpha}\in C^{2,\gamma}(\overline{\mathcal{D}})
\]
and
\[
\mathcal{L}v_{\alpha}-\alpha v_{\alpha}=0
\qquad\text{in }\mathcal{D},
\]
with $v_{\alpha}=1
\quad\text{on }\partial\mathcal{X}_r,
\qquad
v_{\alpha}=0
\quad\text{on }\partial\mathcal{X}$.

Since $\mathcal{X}_0$ is compact and $v_{\alpha}$ is continuous, there
exists $\eta>0$ such that
\[
v_{\alpha}(\bm{x})
\geq
\epsilon+\eta,
\qquad
\forall\bm{x}\in\mathcal{X}_0.
\]

Fix an arbitrary $\beta\geq0$. Choose
\[
0<\delta
<
\left(
1+\frac{\beta}{\alpha}
\right)\eta
\]
and define
\[
q(\bm{x})
:=
\left(
1+\frac{\beta}{\alpha}
\right)
v_{\alpha}(\bm{x})
-
\frac{\beta}{\alpha}
-
\delta.
\]
Then
\[
q\in C^{2,\gamma}(\overline{\mathcal{D}}).
\]
Since
\[
\mathcal{L}v_{\alpha}-\alpha v_{\alpha}=0,
\]
we have
\[
\mathcal{L}q-\alpha q
=
\beta+\alpha\delta
>
\beta,
\qquad
\forall\bm{x}\in\mathcal{D}.
\]
Because $\mathcal{L}q-\alpha q$ is continuous on
$\overline{\mathcal{D}}$, the inequality extends to the closure:
\[
\mathcal{L}q-\alpha q
=
\beta+\alpha\delta
>
\beta,
\qquad
\forall\bm{x}\in\overline{\mathcal{D}}.
\]
Furthermore, on $\partial\mathcal{X}_r$,
\[
q
=
1-\delta
<
1,
\]
while on $\partial\mathcal{X}$,
\[
\alpha q+\beta
=
-\alpha\delta
<
0.
\]
Finally, for every $\bm{x}\in\mathcal{X}_0$,
\[
\begin{aligned}
q(\bm{x})
&\geq
\left(
1+\frac{\beta}{\alpha}
\right)
(\epsilon+\eta)
-
\frac{\beta}{\alpha}
-
\delta\\
&=
\epsilon
+
\frac{\beta}{\alpha}(\epsilon-1)
+
\left(
1+\frac{\beta}{\alpha}
\right)\eta
-
\delta\\
&>
\epsilon
+
\frac{\beta}{\alpha}(\epsilon-1).
\end{aligned}
\]
Thus all four inequalities in \eqref{barrier2} hold strictly for $q$
on the corresponding sets.

It remains to construct a polynomial satisfying the same strict
inequalities. 

By the extension theorem for Hölder spaces, e.g., Lemma 6.37 in
\cite{gilbarg1977elliptic}, for each $\ell$ there exist an open
neighborhood $U_\ell$ of $\overline{\mathcal{D}_\ell}$ and an extension
\[
q_\ell\in C^{2,\gamma}(U_\ell)
\]
such that $q_\ell=q \text{ on }\overline{\mathcal{D}_\ell}$.

Under Assumption \ref{assm2}, we have that the closures $\overline{\mathcal{D}_\ell}$ are pairwise disjoint
compact sets. Thus, the neighborhoods $U_\ell$ can be chosen to be pairwise disjoint.
Hence, with
\[
U:=\bigcup_{\ell=1}^N U_\ell,
\]
the function
\[
\widetilde q(\bm{x})
:=
q_\ell(\bm{x}),
\qquad
\bm{x}\in U_\ell,
\]
is well defined and belongs to $C^{2,\gamma}(U)$. Moreover,
\[
\widetilde q=q
\qquad\text{on }\overline{\mathcal{D}}.
\]
For simplicity, we henceforth denote this extension again by $q$.

Choose a compact set $K$ such that
\[
\overline{\mathcal{D}}\subset K\subset U.
\]
By Lemma~\ref{lem:poly_approx}, there exists a sequence of polynomials
$p_k\in\mathbb{R}[\bm{x}]$ such that
\[
\|p_k-q\|_{C^2(K)} \longrightarrow 0,
\]
where \[
\|f\|_{C^2(K)}:=
\max\left\{
\begin{split}
&\sup_{\bm{x}\in K}|f(\bm{x})|,
\sup_{\bm{x}\in K}\|\nabla f(\bm{x})\|,\\
&\sup_{\bm{x}\in K}\|\bm{H}f(\bm{x})\|_F
\end{split}
\right\}.
\]
In particular, $\|p_k-q\|_{C^2(\overline{\mathcal{D}})}
\longrightarrow 0$.

Let
\[
B:=\sup_{\bm{x}\in\overline{\mathcal{D}}}
\|\bm{b}(\bm{x})\|,
\qquad
A:=\sup_{\bm{x}\in\overline{\mathcal{D}}}
\|\bm{a}(\bm{x})\|_F.
\]
Since $\bm{b}$ and $\bm{a}$ are continuous on the compact set
$\overline{\mathcal{D}}$, we have $A,B<\infty$. For every $k$,
\[
\begin{aligned}
&\|\mathcal{L}p_k-\mathcal{L}q\|_{C(\overline{\mathcal{D}})}\\
\leq &
B\|\nabla p_k-\nabla q\|_{C(\overline{\mathcal{D}})}+\frac{A}{2}
\|\bm{H}p_k-\bm{H}q\|_{C(\overline{\mathcal{D}}),F}
\\
\leq&
\left(B+\frac{A}{2}\right)
\|p_k-q\|_{C^2(\overline{\mathcal{D}})}.
\end{aligned}
\]
Hence,
\[
\|\mathcal{L}p_k-\mathcal{L}q\|_{C(\overline{\mathcal{D}})}
\longrightarrow0.
\]

Because the inequalities satisfied by $q$ are strict and the sets
$\mathcal{X}_0$, $\partial\mathcal{X}_r$, and $\partial\mathcal{X}$
are compact, there exists $\rho>0$ such that
\[
q(\bm{x})
\geq
\epsilon+\frac{\beta}{\alpha}(\epsilon-1)+\rho,
\qquad
\forall\,\bm{x}\in\mathcal{X}_0,
\]
\[
q(\bm{x})\leq 1-\rho,
\qquad
\forall\,\bm{x}\in\partial\mathcal{X}_r,
\]
and
\[
\alpha q(\bm{x})+\beta\leq-\rho,
\qquad
\forall\,\bm{x}\in\partial\mathcal{X}.
\]
Moreover, since
\[
\mathcal{L}q-\alpha q=\beta+\alpha\delta
\qquad\text{on }\overline{\mathcal{D}},
\]
we may, by decreasing $\rho$ if necessary, also assume that
\[
\mathcal{L}q-\alpha q
\geq\beta+\rho,
\qquad
\forall\,\bm{x}\in\overline{\mathcal{D}}.
\]

Choose $k$ sufficiently large such that
\[
\|p_k-q\|_{C^2(\overline{\mathcal{D}})}
<
\min\left\{
\frac{\rho}{2},
\frac{\rho}
{2\left(B+\frac{A}{2}+\alpha\right)}
\right\}.
\]
Then, for every $\bm{x}\in\mathcal{X}_0$,
\[
p_k(\bm{x})
\geq
q(\bm{x})-\frac{\rho}{2}
>
\epsilon+\frac{\beta}{\alpha}(\epsilon-1),
\]
and, for every $\bm{x}\in\partial\mathcal{X}_r$,
\[
p_k(\bm{x})
\leq
q(\bm{x})+\frac{\rho}{2}
<1.
\]
Similarly, for every $\bm{x}\in\partial\mathcal{X}$,
\[
\alpha p_k(\bm{x})+\beta
\leq
\alpha q(\bm{x})+\beta
+\alpha\|p_k-q\|_{C(\overline{\mathcal{D}})}
<0.
\]

Finally, for every $\bm{x}\in\overline{\mathcal{D}}$,
\[
\begin{aligned}
&\mathcal{L}p_k(\bm{x})-\alpha p_k(\bm{x})\\
\geq &
\mathcal{L}q(\bm{x})-\alpha q(\bm{x})-
\left\{
\|\mathcal{L}p_k-\mathcal{L}q\|_{C(\overline{\mathcal{D}})}
+
\alpha\|p_k-q\|_{C(\overline{\mathcal{D}})}
\right\}
\\
>&
\beta.
\end{aligned}
\]
Therefore, setting
\[
v:=p_k\in\mathbb{R}[\bm{x}],
\]
we obtain a polynomial satisfying all four inequalities in
\eqref{barrier2} strictly. This completes the proof.
\end{proof}

\section{Semidefinite Programming Computation}
\label{sec:sdp}

In this section, we present the synthesis of polynomial barrier functions using SOS programming for polynomial systems, i.e., systems in which each component of the drift vector and each entry of the diffusion matrix is a polynomial, and the initial set, the safe set, and the target set are semialgebraic. We replace polynomial positivity conditions with membership in the corresponding Archimedean quadratic modules to formulate the SOS program. Putinar's Positivstellensatz guarantees that strictly positive polynomials admit such representations. This yields a complete degree hierarchy for polynomial barrier functions.

\subsection{Algebraic Representations}

\begin{assumption}
\label{ass:poly}
\begin{enumerate}
\item We assume that the four state-space regions appearing in the barrier-like condition \eqref{barrier2} admit compact basic semialgebraic representations:
\[
\begin{cases}
\mathcal{X}_0=
\left\{
\bm{x}\in\mathbb{R}^n
\mid
g_{0,i}(\bm{x})\geq0,\;
i=1,\ldots,k_0
\right\},\\
\overline{\mathcal{X}\setminus\mathcal{X}_r}=
\left\{
\bm{x}\in\mathbb{R}^n
\mid
g_{u,i}(\bm{x})\geq0,\;
i=1,\ldots,k_u
\right\},\\
\partial \mathcal{X}_r=
\left\{
\bm{x}\in\mathbb{R}^n
\mid
g_{r,i}(\bm{x})\geq0,\;
i=1,\ldots,k_r
\right\},\\
\partial \mathcal{X}=
\left\{
\bm{x}\in\mathbb{R}^n
\mid
g_{e,i}(\bm{x})\geq0,\;
i=1,\ldots,k_e
\right\},
\end{cases}
\]
where $g_{0,i}, g_{u,i}, g_{r,i}, g_{e,i} \in\mathbb{R}[\bm{x}]$ 
are known polynomials.
\item the components of $\bm{b}$ and the entries of $\bm{\sigma}$ are polynomial in $\bm{x}$.
\end{enumerate}
\end{assumption}

We next introduce the SOS framework used to certify positivity
conditions. Let
\[
\Sigma[\bm{x}]:=
\left\{
q\in\mathbb{R}[\bm{x}]
\mid
q(\bm{x})=\sum_j h_j^2(\bm{x}),
\quad
h_j\in\mathbb{R}[\bm{x}]
\right\}
\]
denote the cone of SOS polynomials. For a collection
of defining polynomials
\[
G=\{g_1,\ldots,g_k\},
\]
the associated quadratic module is
\[
\mathcal{M}(G):=
\left\{
\sigma_0(\bm{x})+ \sum_{i=1}^{k}
\sigma_i(\bm{x})g_i(\bm{x})
\;\middle|\;
\sigma_0,\sigma_1,\ldots,\sigma_k\in\Sigma[\bm{x}]
\right\}.
\]

To apply Putinar's Positivstellensatz, we ensure that each of the
quadratic modules associated with the four regions is Archimedean, i.e., there exists $R>0$ such that
\[R-\|\bm{x}\|_2^2 \in \mathcal{M}(G).\]
Because $\overline{\mathcal{X}}$ is compact, there definitely exist $R>0$ such that
\[
\overline{\mathcal{X}}
\subseteq
\left\{
\bm{x}\in\mathbb{R}^n
\mid
R-\|\bm{x}\|_2^2\geq0
\right\}.
\]
Define the redundant ball constraint
\[
g_{\mathrm b}(\bm{x}):=R-\|\bm{x}\|_2^2.
\]
Since each of the four regions is contained in
$\overline{\mathcal{X}}$, this constraint does not change any of the
sets. We therefore include $g_{\mathrm b}$ among the defining
polynomials of each region. 
For example, let
\[
G_0=\{g_{0,1},\ldots,g_{0,k_0},g_{\mathrm b}\},
\qquad
G_u =\{g_{u,1},\ldots,g_{u,k_u},g_{\mathrm b}\},
\]
and similarly define
\[
G_r=\{g_{r,1},\ldots,g_{r,k_r},g_{\mathrm b}\},
\qquad
G_e=\{g_{e,1},\ldots,g_{e,k_e},g_{\mathrm b}\}.
\]
The associated quadratic modules
$\mathcal{M}(G_0)$, $\mathcal{M}(G_u)$,
$\mathcal{M}(G_r)$, and $\mathcal{M}(G_e)$ are therefore
Archimedean.

The following form of Putinar's Positivstellensatz
\cite{putinar1993positive} provides the basis for our SOS
formulation.

\begin{theorem}[Putinar's Positivstellensatz \cite{putinar1993positive}]
\label{thm:putinar}
Let
\[
G=
\left\{
\bm{x}\in\mathbb{R}^n
\;\middle|\;
g_i(\bm{x})\geq0,\;
i=1,\ldots,k
\right\}
\]
be a basic closed semialgebraic set. Suppose that $\mathcal{M}(G)$ is Archimedean. Then, for every polynomial
$p\in\mathbb{R}[\bm{x}]$ satisfying $p(\bm{x})>0, \forall\,\bm{x}\in G$, there exist SOS polynomials
$\sigma_0,\sigma_1,\ldots,\sigma_k\in\Sigma[\bm{x}]$ such that
\[
p(\bm{x})=\sigma_0(\bm{x})+\sum_{i=1}^{k} \sigma_i(\bm{x})g_i(\bm{x}).
\]
Equivalently, $p\in\mathcal{M}(G)$.
\end{theorem}


Thus, for a fixed polynomial degree and fixed multiplier degree, the resulting SOS conditions lead to a finite-dimensional semi-definite program (SDP). Moreover, since the corresponding quadratic modules are Archimedean, Putinar's Positivstellensatz guarantees that every polynomial that is strictly positive on one of these regions admits an SOS representation at some finite degree. Consequently, the SOS hierarchy is complete in the sense that every strictly positive polynomial constraint is certified at a finite hierarchy level. They are formally shown in the sequel. The above results have also been widely applied in the existing reachability analysis literature, e.g., \cite{henrion2013convex,xue2019inner}.

\subsection{SOS Reformulation and Soundness}

We now translate the four inequalities in the barrier-like condition \eqref{barrier2}  into SOS constraints. 

Fix $\alpha>0$ and a prescribed probability level
$\epsilon\in (0,1)$. The scalar $\beta\geq0$ is treated as a decision
variable. Define the four residual polynomials
\begin{equation}
\label{eq:residuals}
\begin{aligned}
p_1(\bm{x})
&:=
v(\bm{x})
-\epsilon
-\frac{\beta}{\alpha}(\epsilon-1),\\
p_2(\bm{x})
&:=
\mathcal Lv(\bm{x})
-\alpha v(\bm{x})
-\beta,\\
p_3(\bm{x})
&:=
1-v(\bm{x}),\\
p_4(\bm{x})
&:=
-\alpha v(\bm{x})-\beta.
\end{aligned}
\end{equation}


We enforce the four inequalities by requiring each residual to belong to
the corresponding quadratic module:
\begin{equation}
\label{putin_sos}
\begin{cases}
p_1(\bm{x})=
\sigma_{0,0}(\bm{x})+\sum_{i=1}^{k_0+1}
\sigma_{0,i}(\bm{x})g_{0,i}(\bm{x}),\\[1.2ex]
p_2(\bm{x})=\sigma_{D,0}(\bm{x})+\sum_{j=1}^{k_u+1}
\sigma_{D,j}(\bm{x})g_{D,j}(\bm{x}),\\[1.2ex]
p_3(\bm{x})=\sigma_{r,0}(\bm{x})+ \sum_{\ell=1}^{k_r+1}
\sigma_{r,\ell}(\bm{x})g_{r,\ell}(\bm{x}),\\[1.2ex]
p_4(\bm{x})=\sigma_{X,0}(\bm{x})+
\displaystyle\sum_{s=1}^{k_e+1}
\sigma_{X,s}(\bm{x})g_{X,s}(\bm{x}),
\end{cases}
\end{equation}
which is equivalent to
\begin{equation}
\label{SOS}
\begin{cases}
p_1(\bm{x})-
\displaystyle\sum_{i=1}^{k_0+1}
\sigma_{0,i}(\bm{x})g_{0,i}(\bm{x}) \in \sum[\bm{x}],\\[1.2ex]
p_2(\bm{x})-
\displaystyle\sum_{j=1}^{k_u+1}
\sigma_{D,j}(\bm{x})g_{D,j}(\bm{x}) \in \sum[\bm{x}],\\[1.2ex]
p_3(\bm{x})-
\displaystyle\sum_{\ell=1}^{k_r+1}
\sigma_{r,\ell}(\bm{x})g_{r,\ell}(\bm{x}) \in \sum[\bm{x}],\\[1.2ex]
p_4(\bm{x})-
\displaystyle\sum_{s=1}^{k_e+1}
\sigma_{X,s}(\bm{x})g_{X,s}(\bm{x}) \in \sum[\bm{x}].
\end{cases}
\end{equation}
Here, all multiplier polynomials are SOS, and $g_{\ell,k_\ell+1}=g_{\mathrm{b}}, \ell\in\{0,u,r,e\}$.

\begin{theorem}[Soundness of the SOS certificate]
\label{thm:sos_soundness}
Suppose that Assumptions \ref{assump:sde_conditions} and
\ref{assm2} hold. Fix $\alpha>0$ and a prescribed
$\epsilon\in (0,1)$. If a polynomial $v\in \mathbb{R}[\bm{x}]$, a scalar $\beta\geq0$, and SOS
multipliers satisfy \eqref{SOS}, then $v$ satisfies the  barrier-like condition \eqref{barrier2}.
Consequently, $\mathbb{P}_{\mathrm{RA}}(\bm{x}_0)
\geq \epsilon, \forall\bm{x}_0\in\mathcal X_0$.
\end{theorem}

\begin{proof}
Every SOS polynomial is globally nonnegative, and every defining
polynomial is nonnegative on its associated semialgebraic set.
Therefore, each right-hand side of \eqref{SOS} is nonnegative on the
corresponding set. Hence
\[
p_1\geq0
\quad\text{on }\mathcal X_0,
\qquad
p_2\geq0
\quad\text{on }\overline{\mathcal{X}\setminus \mathcal{X}_r},
\]
and
\[
p_3\geq0
\quad\text{on }\partial\mathcal X_r,
\qquad
p_4\geq0
\quad\text{on }\partial\mathcal X.
\]
In particular, the barrier-like condition
\eqref{barrier2} holds. The reach-avoid guarantee then follows from
Theorem~\ref{thm:poly_sbc}.
\end{proof}




\subsection{Completeness of the SOS Formulation}

In this section, we show that the SOS hierarchy is complete over
polynomial barrier functions. The argument combines the
strict-inequality result established in Theorem \ref{thm:poly_sbc}
with Putinar's Positivstellensatz.

\begin{theorem}[Completeness of the SOS hierarchy]
\label{cor:sos_complete}
Suppose Assumptions \ref{assump:sde_conditions}--\ref{ass:poly} hold. If the reach-avoid probability is strictly larger than $\epsilon$ for every state in the initial set $\mathcal{X}_0$, i.e., $\mathbb{P}_{\mathrm{RA}}(\bm{x}_0)>\epsilon, \forall\bm{x}_0\in\mathcal{X}_0$, there exist a polynomial $v\in \mathbb{R}[\bm{x}]$, $\alpha>0$, and $\beta\geq 0$ such that the SOS program \eqref{SOS} hold for some finite-degree SOS multipliers.
\end{theorem}
\begin{proof}
By Theorem \ref{thm:poly_sbc}, there
exist a polynomial $v \in \mathbb{R}[\bm{x}]$, $\beta\geq 0$, and $\alpha>0$ such that the four barrier-like
inequalities in \eqref{barrier2} hold strictly on the
compact sets on which the corresponding residuals are imposed.
Therefore, the residual polynomials in \eqref{eq:residuals} satisfy
\[
p_1>0
\quad\text{on }\mathcal X_0,
\qquad
p_2>0
\quad\text{on }\overline{\mathcal D},
\]
and
\[
p_3>0
\quad\text{on }\partial\mathcal X_r,
\qquad
p_4>0
\quad\text{on }\partial\mathcal X.
\]

Since the four corresponding quadratic modules are Archimedean,
Putinar's Positivstellensatz applies separately to the four regions.
Hence there exist SOS polynomials \[
\left\{
\begin{split}
&\sigma_{0,0},\ldots,\sigma_{0,k_0+1},\sigma_{D,0},\ldots,\sigma_{D,k_u+1},\\
&\sigma_{r,0},\ldots,\sigma_{r,k_r+1},\sigma_{X,0},\ldots,\sigma_{X,k_e+1}
\end{split}\right\}
\] such that \eqref{putin_sos} holds. \eqref{putin_sos} is  equivalent to \eqref{SOS}. Each representation contains finitely many SOS polynomials and therefore has finite degree. Consequently, there exists a finite degree level at which the corresponding SOS program \eqref{SOS} is feasible.
\end{proof}


\section{Illustrative Examples}
\label{sec:ex}
In this section, we demonstrate the application of our theoretical developments through two numerical examples with polynomial systems. We aim to find polynomial barrier functions that satisfy the SOS constraints \eqref{SOS}. For a fixed $\alpha>0$, the resulting SOS program is a convex SDP in the coefficients of the polynomial $p$,
the scalar $\beta$, and the Gram matrices of the SOS multipliers. The parameter $\alpha$ therefore determines the convex SDP subproblem to be solved. As shown in Lemma \ref{lemma:stri} and the proof of Theorem \ref{thm:poly_sbc}, the barrier-like condition becomes less conservative as $\alpha$ approaches $0$. This suggests a practical computational strategy: fix $\alpha$ and solve a sequence of convex SDPs for values of $\alpha$ successively closer to $0$. In the experiments, the SOS programs are formulated using the MATLAB-based modeling toolbox YALMIP \cite{lofberg2004yalmip} and the resulting SDPs are solved using MOSEK 10.1.21 \cite{aps2019mosek}. To ensure numerical stability during the solution of these SDPs, we additionally impose a constraint on the coefficients of the unknown polynomials, specifically restricting them to the interval $[-10^2, 10^2]$.

\begin{example}
\label{ex1}
Consider the one-dimensional SDE
\[
dX(t)=-X(t)dt+\frac{\sqrt{2}}{2}X(t)dW(t).
\]
This model describes a simple population-dynamics system subject to multiplicative environmental fluctuations. We consider the compact initial set $\mathcal{X}_0=\left\{
x\in\mathbb{R} \mid 0.01-(x-0.8)^2\geq 0
\right\}$,
the safe set $\mathcal{X}=\left\{
x\in\mathbb{R} \mid 1-x^2>0
\right\}$, and the compact target set
$\mathcal{X}_r=\left\{
x\in\mathbb{R} \mid 1-100x^2\geq0
\right\}$. Thus, $\overline{\mathcal{D}}=\overline{\mathcal{X}\setminus\mathcal{X}_r}= \{x\mid 1-x^2\geq 0, x^2-0.01\geq 0\}$.
For this system,
\[
b(x)=-x,
\qquad
\sigma(x)=\frac{\sqrt{2}}{2}x,
\qquad
a(x)=\sigma(x)^2=\frac{x^2}{2}.
\]
Both $b$ and $a$ are smooth. Moreover, $a(x)\geq\frac{1}{200}>0, \forall x\in \overline{\mathcal{D}}$,
so the diffusion is uniformly elliptic on the relevant continuation
component. The boundary sets are
\[
\partial\mathcal{X}=\{x\mid x^2-1=0\},
\qquad
\partial\mathcal{X}_r=\{x\mid x^2-0.01=0\},
\]
and hence $\partial\mathcal{X}\cap\partial\mathcal{X}_r=\emptyset$.
Thus, the regularity and ellipticity assumptions required by the proposed Dirichlet characterization are satisfied.

Since the initial states are positive and the solution preserves the sign of the initial state, reaching the target is equivalent to reaching the boundary $x=0.1$. Therefore, for 
$x_0\in\mathcal{X}_0$, 
\[ 
\mathbb{P}_{\mathrm{RA}}(x_0) = \mathbb{P}(\tau_{0.1}<\tau_1), 
\] where $\tau_{0.1}$ and $\tau_1$ denote the first hitting times of $0.1$ and $1$, respectively.  Consequently, for $x_0\in(0.1,1)$, 
\[ 
\mathbb{P}_{\mathrm{RA}}(x_0) = \frac{s(1)-s(x_0)} {s(1)-s(0.1)} = \frac{1-x_0^5}{1-0.1^5}, 
\] 
where $s(x):=x^5$.
Since this expression is decreasing in $x_0$, 
\[ 
\min_{x_0\in\mathcal{X}_0}\mathbb{P}_{\mathrm{RA}}(x_0) = \mathbb{P}_{\mathrm{RA}}(0.9) = \frac{1-0.9^5}{1-0.1^5} \approx 0.4095. 
\] 
The proof of the above statement is shown in Lemma \ref{lem:ex1_worstcase} in Appendix. This value provides a benchmark for the certified lower bounds obtained from the SOS program \eqref{SOS}.

 The computed feasibility results are reported in Table \ref{tab:sdpex1_set} with $g_{\mathrm{b}}=1-x^2$.
 
\begin{table}[H]
\caption{\centering Feasibility of SDP \eqref{SOS} for Example \ref{ex1}\\
(\ding{52}: feasible; \ding{55}: infeasible)}
\centering
\begin{tabular}{|c|c|c|c|c|c|}
\hline
Degree&$\alpha$ & $\epsilon=0.3$ & $\epsilon=0.35$ & $\epsilon=0.4$  &$\epsilon=0.408$ 
\\
\hline
4 &0.1  & \ding{55} & \ding{55} & \ding{55} & \ding{55} \\
4 &0.01  & \ding{52} & \ding{52} & \ding{55} & \ding{55} \\
4 &0.001  & \ding{52} & \ding{52} & \ding{55} & \ding{55} \\
6 &0.1  & \ding{55} & \ding{55} & \ding{55} & \ding{55} \\
6 &0.01  & \ding{52} & \ding{52} & \ding{55} & \ding{55} \\
6 &0.001  & \ding{52} & \ding{52} & \ding{52} & \ding{55} \\
8 &0.1  & \ding{52} & \ding{55} & \ding{55} & \ding{55} \\
8 &0.01  & \ding{52} & \ding{52} & \ding{52} & \ding{55} \\
8 &0.001  & \ding{52} & \ding{52} & \ding{52} & \ding{52} \\
10 &0.1  & \ding{55} & \ding{55} & \ding{55}& \ding{55} \\
10 &0.01  & \ding{52} & \ding{52} & \ding{52}& \ding{55} \\
10 &0.001  & \ding{52} & \ding{52} & \ding{52}& \ding{52} \\
12 &0.1  & \ding{52} & \ding{55} & \ding{55} & \ding{55} \\
12 &0.01  & \ding{52} & \ding{52} & \ding{52} & \ding{55} \\
12 &0.001  & \ding{52} & \ding{52} & \ding{52} & \ding{52} \\
14 &0.1  & \ding{52} & \ding{55} & \ding{55} & \ding{55}\\
14 &0.01  & \ding{52} & \ding{52} & \ding{52} & \ding{55}\\
14 &0.001  & \ding{52} & \ding{52} & \ding{52} & \ding{52}\\
\hline
\end{tabular}
\label{tab:sdpex1_set}
\end{table}

\vspace{0.2cm}

\end{example}

\begin{example}
\label{ex2}

Consider the two-dimensional Ornstein--Uhlenbeck process
\[
d\bm{X}_t=-\bm{X}_t\,dt+\sigma\,d\bm W_t,
\qquad
\bm{X}_t\in\mathbb{R}^2,
\]
with $\sigma=0.3$. Let the safe, target, and initial sets be defined, respectively, as  $\mathcal{X}=
\{\bm{x}\in\mathbb{R}^2\mid 1-\|\bm{x}\|>0\}$, $\mathcal{X}_r = \{\bm{x}\in\mathbb{R}^2\mid 0.25-\|\bm{x}\|\geq 0\}$, and   $\mathcal{X}_0 =\{\bm{x} \in \mathbb{R}^2\mid (x-0.75)^2+y^2=0\}$. Thus,  $\mathcal D=\mathcal{X}\setminus\mathcal{X}_r=
\{\bm{x}\in\mathbb{R}^2 \mid \|\bm{x}\|-0.25>0, 1-\|\bm{x}\|>0\}$.

All standing assumptions are satisfied. The drift
$\bm b(\bm{x})=-\bm{x}$ is globally Lipschitz, and
the diffusion matrix $\bm a(\bm{x})=\sigma^2 I_2$ is uniformly elliptic with ellipticity constant $\lambda=\sigma^2>0$.
Moreover,
\[
\partial\mathcal{X}
=
\{\bm{x} \mid \|\bm{x}\|-1=0\},
\qquad
\partial\mathcal{X}_r
=
\{\bm{x} \mid \|\bm{x}\|-0.25=0\}
\]
are disjoint $C^\infty$ curves. Hence, $\mathcal D$ is a bounded
$C^{2,\gamma}$ domain for every $\gamma\in(0,1)$. The coefficients $\bm b$ and $\bm a$ are polynomial, so this example is also compatible with the SOS formulation.

By rotational symmetry, the reach-avoid probability depends only on the radial coordinate $\rho=\|\bm{x}\|$. Thus,
\[\mathbb{P}_{\mathrm{RA}}(\bm{x})=P(\rho),\] 
where 
\[
P(\rho)
=
\frac{
\displaystyle
\int_{\rho}^{1}
\frac{1}{s}
\exp\left(\frac{s^2}{\sigma^2}\right)\,ds
}{
\displaystyle
\int_{0.25}^{1}
\frac{1}{s}
\exp\left(\frac{s^2}{\sigma^2}\right)\,ds
}.
\]
The proof of the above statement is shown in Lemma \ref{lem:ou2d_exact} in Appendix.

For $\sigma=0.3$ and $\bm{x}_0=(0.75,0)$, $\mathbb{P}_{\mathrm{RA}}(0.75,0) \approx 0.9846$. Thus, the strict-margin assumption holds for every $0\leq\epsilon<0.9846$. In the numerical
experiments, we consider $\epsilon\in\{0.9, 0.95, 0.98\}$. The computed feasibility results are reported in Table \ref{tab:sdpex2_set} with $g_{\mathrm{b}}=1-\|\bm{x}\|^2$.

\begin{table}[H]
\caption{\centering Feasibility of SDP \eqref{SOS} for Example \ref{ex2}\\
(\ding{52}: feasible; \ding{55}: infeasible)}
\centering
\begin{tabular}{|c|c|c|c|c|}
\hline
Degree&$\alpha$ & $\epsilon=0.9$ & $\epsilon=0.95$ & $\epsilon=0.98$ 
\\
\hline
4 &0.1  & \ding{55} & \ding{55} & \ding{55}  \\
4 &0.01  & \ding{55} & \ding{55} & \ding{55}  \\
4 &0.0001  & \ding{55} & \ding{55} & \ding{55} \\
6 &0.1  & \ding{55} & \ding{55} & \ding{55}  \\
6 &0.01  & \ding{55} & \ding{55} & \ding{55} \\
6 &0.0001  & \ding{55} & \ding{55} & \ding{55}  \\
8 &0.1  & \ding{55} & \ding{55} & \ding{55}  \\
8 &0.01  & \ding{52} & \ding{55} & \ding{55} \\
8 &0.0001  & \ding{52} & \ding{55} & \ding{55} \\
10 &0.1  & \ding{55} & \ding{55} & \ding{55}\\
10 &0.01  & \ding{52} & \ding{52} & \ding{55} \\
10 &0.0001  & \ding{52} & \ding{52} & \ding{55} \\
12 &0.1  & \ding{55} & \ding{55} & \ding{55} \\
12 &0.01  & \ding{52} & \ding{52} & \ding{55}\\
12 &0.0001  & \ding{52} & \ding{52} & \ding{55} \\
14 &0.1  & \ding{55} & \ding{55} & \ding{55} \\
14 &0.01  & \ding{52} & \ding{52} & \ding{55} \\
14 &0.0001  & \ding{52} & \ding{52} & \ding{52} \\
16 &0.1  & \ding{55} & \ding{55} & \ding{55} \\
16 &0.01  & \ding{52} & \ding{52} & \ding{55} \\
16 &0.0001  & \ding{52} & \ding{52} & \ding{52} \\
18 &0.1  & \ding{55} & \ding{55} & \ding{55} \\
18 &0.01  & \ding{52} & \ding{52} & \ding{55} \\
18 &0.0001  & \ding{52} & \ding{52} & \ding{52} \\
20 &0.1  & \ding{55} & \ding{55} & \ding{55} \\
20 &0.01  & \ding{52} & \ding{52} & \ding{55} \\
20 &0.0001  & \ding{52} & \ding{52} & \ding{52} \\
\hline
\end{tabular}
\label{tab:sdpex2_set}
\end{table}
\vspace{0.2cm}
\end{example}

The results in Tables \ref{tab:sdpex1_set} and \ref{tab:sdpex2_set}
illustrate the behavior of the proposed SOS-based verification
framework from complementary perspectives. First, for a fixed polynomial degree, decreasing $\alpha$ generally improves feasibility. This effect is evident in both examples. In Example \ref{ex1}, for instance, with polynomial degree $8$, the SDP is infeasible for $\epsilon=0.4$ when $\alpha=0.1$, whereas it becomes feasible for $\alpha=0.01$ and $\alpha=0.001$. A similar trend is observed in Example \ref{ex2}: for polynomial degrees $8$--$20$, the cases with $\alpha=0.01$ or $\alpha=0.0001$ are substantially more favorable than those with $\alpha=0.1$. This behavior is consistent with Lemma \ref{lemma:stri}, which shows that the barrier-like condition becomes less conservative as $\alpha\to0$. Second, increasing the polynomial degree generally improves the ability of the SOS relaxation to recover feasible certificates. In Example \ref{ex1}, when $\alpha=0.001$, the  threshold $\epsilon=0.408$ becomes feasible at degree $8$. Since the exact worst-case reach-avoid probability over the initial set is approximately $0.4095$, the resulting certified lower bound is close to the exact value. In Example \ref{ex2}, the same phenomenon is particularly clear near the largest probability threshold. For $\alpha=0.0001$, the SDP is feasible for $\epsilon=0.9$ starting at degree $8$, for $\epsilon=0.95$ starting at degree $10$, and for
$\epsilon=0.98$ starting at degree $14$. Thus, increasing the
polynomial degree systematically allows the SOS relaxation to recover certificates for increasingly tight probability bounds.

Example \ref{ex2} also demonstrates that the difficulty of the
verification problem depends strongly on the target probability level. For $\sigma=0.3$, the exact reach-avoid probability at the initial state
is $\mathbb{P}_{\mathrm{RA}}(\bm{x}_0) \approx 0.9846$.
Consequently, the threshold $\epsilon=0.9$ leaves a relatively large margin below the exact value, and feasible certificates are obtained already at degree $8$ for sufficiently small $\alpha$. In contrast, the threshold $\epsilon=0.98$ is much closer to the exact probability, and certificates require degree $14$ for $\alpha=0.0001$. This illustrates the expected behavior that tighter certified probability bounds generally require richer polynomial certificates. The results of Example \ref{ex2} further illustrate the interaction between the discount parameter and the polynomial degree. For example, at degree $14$, $\epsilon=0.95$ is feasible for both $\alpha=0.01$ and $\alpha=0.0001$, whereas $\epsilon=0.98$ is feasible only for $\alpha=0.0001$. This indicates that reducing $\alpha$ can compensate, to some extent, for the conservatism of the discounted barrier-like conditions. Nevertheless, the effect of $\alpha$ and the effect of polynomial degree are distinct: reducing $\alpha$ improves the
underlying barrier-like condition, while increasing the polynomial degree improves the approximation capability of the finite-dimensional SOS representation. The two examples also highlight the quantitative tightness that can be achieved by the proposed approach. In Example \ref{ex1}, the certified bound $\epsilon=0.408$ is within approximately $1.5\times10^{-3}$ of the exact worst-case probability $0.4095$. In Example \ref{ex2}, the bound $\epsilon=0.98$ is within approximately $4.55\times10^{-3}$ of the exact probability $0.9846$. Thus, in both examples, the SOS certificates can approach the exact verification threshold from below as the discount parameter is reduced and the polynomial degree is increased. Conversely, whenever the SOS program is solved with a valid numerical certificate, feasibility provides a certified lower bound $\mathbb{P}_{\mathrm{RA}}(\bm{x}_0)\geq\epsilon, \forall \bm{x}_0\in\mathcal{X}_0$, subject to the numerical accuracy of the SDP solution. 

Overall, the numerical experiments demonstrate the trade-off among the discount parameter $\alpha$, the polynomial degree, and the tightness of the certified probability bound. Smaller values of $\alpha$ reduce the conservatism of the resulting barrier-like condition, while higher polynomial degrees provide greater expressive power for representing barrier functions. As illustrated by both examples, these two mechanisms
together enable the SOS hierarchy to obtain certified lower bounds that approach the corresponding exact reach-avoid probabilities. This behavior is consistent with the completeness result for the polynomial barrier hierarchy under the stated Archimedean assumptions.

\section{Conclusion}
\label{sec:con}
We studied infinite-horizon reach-avoid verification for continuous-time stochastic systems modeled by SDEs under suitable regularity and uniform ellipticity assumptions. We established a necessary and sufficient barrier-like condition in terms of polynomial barrier functions for infinite-horizon reach-avoid verification. We further developed an SOS programming formulation with a complete degree hierarchy for synthesizing polynomial barrier functions for polynomial systems. Finally, we demonstrated the theoretical results on two numerical examples.

Future work will focus on relaxing the regularity and uniform ellipticity assumptions to extend the framework to more general SDEs. We also plan to investigate sufficient and necessary barrier-like conditions for finite-horizon safety and reach-avoid verification of SDEs, as well as develop scalable SOS methods for higher-dimensional systems.

\section*{ACKNOWLEDGMENT}
OpenAI’s GPT-5.6 Luna (free version) \cite{chang2026chatgpt} and Deepseek (free version) \cite{liu2024deepseek} were primarily used to polish the language of this paper prior to submission. No methods or computational results presented in this work were generated by AI models. The authors take full responsibility for the accuracy and integrity of all content reported herein.

\section*{References}
\bibliographystyle{IEEEtran}
\bibliography{ref}

\section{Appendix}
\label{subsec:1d_exact}

\begin{lemma}
\label{lem:ex1_worstcase}
In Example~\ref{ex1},
\[
\min_{x_0\in\mathcal{X}_0}
\mathbb{P}_{\mathrm{RA}}(x_0)
=
\mathbb{P}_{\mathrm{RA}}(0.9).
\]
\end{lemma}

\begin{proof}
For $x_0\in\mathcal{X}_0=[0.7,0.9]$, the explicit solution of the
SDE in Example~\ref{ex1} is
\[
X_t=x_0\exp\left(
-\frac54t+\frac{\sqrt{2}}{2}W_t
\right).
\]
Hence
\[
X_t>0,
\qquad
\forall t\ge0,
\quad\text{a.s.}
\]
Since
\[
\mathcal{X}_r=[-0.1,0.1],
\qquad
\mathcal{X}=(-1,1),
\]
a trajectory starting from $x_0\in[0.7,0.9]$ can reach the target only
by hitting $0.1$ and can leave the safe set only by hitting $1$.
Therefore,
\[
\mathbb{P}_{\mathrm{RA}}(x)
=
\mathbb{P}(\tau_{0.1}<\tau_1),
\qquad
x\in(0.1,1).
\]

Define
\[
\tau:=\tau_{0.1}\wedge\tau_1
\]
and consider the function
\[
s(x):=x^5.
\]
The infinitesimal generator of $X_t$ is
\[
\mathcal{L}f(x)=-xf'(x)+
\frac{x^2}{4}f''(x).
\]
Since
\[
s'(x)=5x^4,
\qquad
s''(x)=20x^3,
\]
we have
\[
\mathcal{L}s(x)=-x(5x^4)+
\frac{x^2}{4}(20x^3)=0.
\]
Thus, by It\^o's formula,
\[
s(X_{t\wedge\tau})=s(x)+
\int_0^{t\wedge\tau}
s'(X_u)\frac{\sqrt{2}}{2}X_u\,dW_u.
\]
Since $X_{t\wedge\tau}\in[0.1,1]$, the integrand is bounded on
$[0,\tau]$. Hence the stochastic integral is a martingale, and
\[
\mathbb{E}[s(X_{t\wedge\tau})]=s(x).
\]

Because the process is confined to the compact interval $[0.1,1]$
before $\tau$ and $a$ is uniformly elliptic on $[0.1,1]$,  Proposition 10.1 in \cite{baldi2017stochastic} gives
\[
\tau<\infty
\qquad\text{a.s.}
\]
Therefore,
\[
X_{t\wedge\tau}\longrightarrow X_\tau
\qquad\text{a.s. as }t\to\infty.
\]
Since $0.1\le X_{t\wedge\tau}\le1$,
dominated convergence yields
\[
s(x)
=
\mathbb{E}[s(X_\tau)].
\]
At the stopping time $\tau$,
\[
X_\tau=
\begin{cases}
0.1, & \tau_{0.1}<\tau_1,\\
1, & \tau_1<\tau_{0.1},
\end{cases}
\qquad\text{a.s.}
\]
Since the sample paths are continuous and the two boundary points are
distinct, the probability of simultaneous hitting is zero. Hence
\[
\begin{aligned}
s(x)
&=
s(0.1)
\mathbb{P}(\tau_{0.1}<\tau_1)
+
s(1)
\mathbb{P}(\tau_1<\tau_{0.1})\\
&=
(0.1)^5
\mathbb{P}(\tau_{0.1}<\tau_1)
+
1-
\mathbb{P}(\tau_{0.1}<\tau_1).
\end{aligned}
\]
Solving for the hitting probability gives
\[
\mathbb{P}(\tau_{0.1}<\tau_1)
=
\frac{1-x^5}{1-(0.1)^5},
\qquad
x\in(0.1,1).
\]
Therefore,
\[
\mathbb{P}_{\mathrm{RA}}(x)
=
\frac{1-x^5}{1-(0.1)^5}.
\]
Moreover,
\[
\frac{d}{dx}\mathbb{P}_{\mathrm{RA}}(x)
=
-\frac{5x^4}{1-(0.1)^5}
<0,
\qquad
x\in(0.1,1).
\]
Thus $\mathbb{P}_{\mathrm{RA}}$ is strictly decreasing on
$\mathcal{X}_0=[0.7,0.9]$, and consequently
\[
\min_{x_0\in\mathcal{X}_0}
\mathbb{P}_{\mathrm{RA}}(x_0)
=
\mathbb{P}_{\mathrm{RA}}(0.9).
\]
We complete the proof.
\end{proof}

\begin{lemma}
\label{lem:ou2d_exact}
In Example~\ref{ex2},
\[
\mathbb{P}_{\mathrm{RA}}(0.75,0)=P(0.75).
\]
\end{lemma}

\begin{proof}
Let
\[
R_t:=\|\bm X_t\|.
\]
Since the drift $\bm b(\bm x)=-\bm x$ and the diffusion matrix
$\sigma I_2$ are rotationally invariant, and both $\mathcal X$ and
$\mathcal X_r$ are rotationally symmetric, the reach-avoid probability
depends only on the radial coordinate. Hence, there exists a function
$P:(0.25,1)\to\mathbb R$ such that
\[
\mathbb P_{\mathrm{RA}}(\bm x)=P(\|\bm x\|).
\]

By It\^o's formula applied to $R_t=\|\bm X_t\|$, the radial process is a
one-dimensional diffusion satisfying
\[
dR_t
=
\left(
\frac{\sigma^2}{2R_t}-R_t
\right)dt
+
\sigma\,dB_t,
\qquad
0.25<R_t<1,
\]
where $B_t$ is a one-dimensional Brownian motion. Its infinitesimal
generator is therefore
\[
\mathcal L_R f(\rho)
=
\left(
\frac{\sigma^2}{2\rho}-\rho
\right)f'(\rho)
+
\frac{\sigma^2}{2}f''(\rho).
\]

Let
\[
\tau_{0.25}
:=
\inf\{t\ge0:R_t=0.25\},
\quad
\tau_1
:=
\inf\{t\ge0:R_t=1\}.
\]
Then
\[
P(\rho)
=
\mathbb P(\tau_{0.25}<\tau_1),
\qquad
\rho\in(0.25,1).
\]

A scale function $s$ for the radial diffusion satisfies
\[
\mathcal L_R s=0.
\]
Hence
\[
\frac{\sigma^2}{2}s''(\rho)
+
\left(
\frac{\sigma^2}{2\rho}-\rho
\right)s'(\rho)
=0,
\]
or equivalently,
\[
s''(\rho)
+
\left(
\frac1\rho-\frac{2\rho}{\sigma^2}
\right)s'(\rho)
=0.
\]
Therefore,
\[
s'(\rho)
=
C
\exp\left(
-\int
\left(
\frac1\rho-\frac{2\rho}{\sigma^2}
\right)d\rho
\right)
=
C\frac1\rho
\exp\left(\frac{\rho^2}{\sigma^2}\right).
\]
Since a scale function is determined only up to a positive affine
transformation, we may choose
\[
s'(\rho)
=
\frac{1}{\rho}
\exp\left(\frac{\rho^2}{\sigma^2}\right),
\qquad \rho\in(0.25,1).
\]
Thus, choosing an arbitrary reference point
$\rho_\ast\in(0.25,1)$, we may define
\[
s(\rho)
=
\int_{\rho_\ast}^{\rho}
\frac{1}{r}
\exp\left(\frac{r^2}{\sigma^2}\right)\,dr.
\]
The particular choice of $\rho_\ast$ is immaterial, since different choices differ only by an additive constant and therefore yield the same hitting probabilities.

For a one-dimensional diffusion with scale function $s$, the
probability of hitting $a$ before $b$, starting from
$\rho\in(a,b)$, is
\[
\mathbb P_\rho(\tau_a<\tau_b)
=
\frac{s(b)-s(\rho)}{s(b)-s(a)}.
\]
Please refer to the proof in Lemma \ref{lem:ex1_worstcase}. 

Taking
\[
a=0.25,
\qquad
b=1,
\]
we obtain
\[
P(\rho)
=
\frac{
\displaystyle
\int_\rho^1
\frac1r
\exp\left(\frac{r^2}{\sigma^2}\right)\,dr
}{
\displaystyle
\int_{0.25}^1
\frac1r
\exp\left(\frac{r^2}{\sigma^2}\right)\,dr
}.
\]
Equivalently, using
\[
\frac{d}{dz}\operatorname{Ei}(z)=\frac{e^z}{z},
\]
with $z=r^2/\sigma^2$, we have
\[
\int
\frac1r
\exp\left(\frac{r^2}{\sigma^2}\right)\,dr
=
\frac12
\operatorname{Ei}\left(\frac{r^2}{\sigma^2}\right).
\]
Thus,
\[
P(\rho)
=
\frac{
\operatorname{Ei}\left(\frac1{\sigma^2}\right)
-
\operatorname{Ei}\left(\frac{\rho^2}{\sigma^2}\right)
}{
\operatorname{Ei}\left(\frac1{\sigma^2}\right)
-
\operatorname{Ei}\left(\frac{0.25^2}{\sigma^2}\right)
}.
\]

For $\sigma=0.3$ and $\bm x_0=(0.75,0)$, we have
$\|\bm x_0\|=0.75$, and therefore
\[
\mathbb P_{\mathrm{RA}}(0.75,0)
=
P(0.75)
=
\frac{
\displaystyle
\int_{0.75}^1
\frac1r
\exp\left(\frac{r^2}{0.09}\right)\,dr
}{
\displaystyle
\int_{0.25}^1
\frac1r
\exp\left(\frac{r^2}{0.09}\right)\,dr
}
\]
We complete the proof.
\end{proof}

\end{document}